\documentclass[a4paper,UKenglish,cleveref, autoref, thm-restate]{lipics-v2021-modified}
\pdfoutput=1 

\title{Beyond Absolute Positiveness for Universally Quantified Non-Linear Polynomial Constraints} 

\author{Carsten Fuhs}{Birkbeck, University of London, UK}{c.fuhs@bbk.ac.uk}{https://orcid.org/0009-0007-3697-4383}{} 

\authorrunning{C. Fuhs} 

\Copyright{Carsten Fuhs} 






\category{} 

\relatedversion{} 

\nolinenumbers 

\EventEditors{Florian Frohn and Étienne Payet}
\EventNoEds{2}
\EventLongTitle{21st International Workshop on Termination (WST 2026)}
\EventShortTitle{WST 2026}
\EventAcronym{WST}
\EventYear{2026}
\EventDate{July 25, 2026}
\EventLocation{Lisbon, Portugal}
\EventLogo{}

\newcommand{\minX}{\mu}
\newcommand{\AAA}{\mathcal{A}}

\newcommand{\Sig}{\mathcal{F}}
\newcommand{\Var}{\mathit{\mathcal{V}ar}}

\newcommand{\aprove}{\textsf{AProVE}}

\newcommand{\sym}[1]{\mathsf{#1}}
\newcommand{\Fa}{\sym{a}}
\newcommand{\Fc}{\sym{c}}
\newcommand{\Fcons}{\sym{cons}}
\newcommand{\Ff}{\sym{f}}
\newcommand{\Fg}{\sym{g}}
\newcommand{\Fh}{\sym{h}}
\newcommand{\Fi}{\sym{i}}

\newcommand{\Fs}{\sym{s}}

\newcommand{\polt}[2]{[#2]_{#1}}
\newcommand{\polf}[2]{#2_{#1}}

\begin{document}

\maketitle

\begin{abstract}
Polynomial interpretations from function symbols to natural numbers
induce a prominent class of monotone algebras and corresponding
well-founded orders on terms,
used, e.g., for termination analysis and complexity analysis of term
rewrite systems.
Finding such polynomial interpretations for a given set of
term constraints involves solving a set of
$\exists\forall$ inequalities over the natural numbers.
Conventionally, the absolute positiveness criterion is used
to reduce $\exists\forall$ inequalities to $\exists$ inequalities.
This extended abstract reports on work in progress to go beyond
absolute positiveness, allowing for finding non-linear polynomial
interpretations that were outside the reach of existing techniques.
\end{abstract}

\section{Introduction}

Termination analysis and complexity analysis for term rewrite systems
(TRSs)\footnote{In this extended abstract, I assume familiarity with
  the basics of term rewrite systems and their termination (see, e.g.,
  \cite{baa:nip:98} for an introduction).}
use well-founded orders on terms as a core ingredient.
A practically highly relevant class of such orders is induced by
polynomial interpretations of function symbols to polynomial
functions over $\mathbb{N}$~\cite{lan:75,tur:49}.
Especially for complexity analysis of TRSs~\cite{hir:mos:08,nos:emm:gie:13},
using non-linear polynomials as part of polynomial interpretations is
important to capture algorithmic complexities of corresponding degree.
In practice, polynomial interpretations are tailored to satisfy a set of
given term constraints that represent strict or weak decrease of a rewrite step.
The check whether such a term constraint $s \succ t$ (or $s \succsim t$)
is satisfied by a
given polynomial interpretation can be reduced to the check whether
a resulting polynomial over the variables of $s$ and $t$ is positive
(or non-negative, respectively) for all instantiations of the variables
by natural numbers. The standard check based on the
\emph{absolute positiveness criterion}~\cite{hon:jak:98,con:mar:tom:urb:05}
requires that all coefficients of this polynomial must be
at least non-negative
and the constant addend must be positive (or non-negative, respectively).
In this extended abstract, I argue that this criterion is needlessly
restrictive for non-linear polynomial interpretations.
I propose an explicit case analysis that
treats natural numbers below a certain threshold $\minX \in \mathbb{N}$
as individual special cases
and uses absolute positiveness only from this threshold on, when
the asymptotic monotonicity behaviour of the polynomial has been reached.

This extended abstract is organised as follows:
\autoref{sec:illustrative} explains the idea with the help of a small example.
\autoref{sec:automation} discusses how the approach can be automated in
a termination or complexity analysis tool.
\autoref{sec:limitation} shows why non-linearity is a prerequisite for
usefulness of the criterion.
\autoref{sec:related} discusses related work, and
\autoref{sec:conclusion} concludes with a discussion of future work.

\section{Illustrative Example}
\label{sec:illustrative}

In this section, I recapitulate polynomial interpretations
and illustrate the problem that this extended abstract
aims to tackle as well as
its proposed solution with the help of an example.

\begin{example}
Consider the following (utterly artificial, yet rather
illustrative) rewrite rule:
\begin{gather}
\Fg(\Fs(x),\Fa) \to \Fc(\Fs(\Fs(x)), x, x, x, x, x) \label{illustrative:rule}
\end{gather}
To prove termination of this rewrite rule, we could use a
strictly monotonic polynomial
interpretation to induce a reduction order $\succ$ (a well-founded
strict order on terms that is
closed under contexts and substitutions) for which the
following term constraint holds:\footnote{Of course, many other
  proof techniques to find termination proofs are applicable here
  as well -- this rewrite rule is not meant as a challenge
  for termination analysis tools in general.}
\begin{gather}
\Fg(\Fs(x),\Fa) \succ \Fc(\Fs(\Fs(x)), x, x, x, x, x) \label{illustrative:order}
\end{gather}
To this end, consider the following strictly monotonic polynomial
interpretation $\AAA$ given for all function symbols of our signature
$\Sig = \{ \Fa, \Fs, \Fg, \Fc \}$:
\begin{center}
$\polf{\AAA}{\Fa} = 1 \qquad
\polf{\AAA}{\Fs}(x_1) = 1 + x_1 \qquad
\polf{\AAA}{\Fg}(x_1,x_2) = x_1 + x_1 \cdot x_2 + x_1^2 + x_2$\\[6pt]
$\polf{\AAA}{\Fc}(x_1,x_2,x_3,x_4,x_5,x_6) = x_1 + x_2 + x_3 + x_4 + x_5 + x_6$
\end{center}
For terms,
we interpret term variables as variables over $\mathbb{N}$
and extend the interpretation homomorphically from function symbols to
function applications so that we can interpret terms as polynomial
functions:
$\polt{\AAA}{x} = x$ and
$\polt{\AAA}{f(t_1,\dots,t_k)} =
 \polf{\AAA}{f}(\polt{\AAA}{t_1},\dots,\polt{\AAA}{t_k})$.\footnote{A formal
  treatment of polynomial interpretations
  can be found, e.g., in \cite{con:mar:tom:urb:05}.}
 For example, for the term $t = \Fs(x)$, we
have
$\polt{\AAA}{t} = \polt{\AAA}{\Fs(x)} = \polf{\AAA}{\Fs}(\polt{\AAA}{x}) =
\polf{\AAA}{\Fs}(x) = 1 + x$.
For the term constraint \eqref{illustrative:order},
we need to show that
$\AAA$ makes the following inequality valid over $\mathbb{N}$:
\begin{gather}
\polt{\AAA}{\Fg(\Fs(x),\Fa)}
  > \polt{\AAA}{\Fc(\Fs(\Fs(x)), x, x, x, x, x)} \label{illustrative:polt}
\end{gather}
By interpreting the terms using $\AAA$, we obtain the following
inequality with implicit universal quantification of the variable $x$:
\begin{gather}
4 + 4 \cdot x + x^2 > 2 + 6 \cdot x
\end{gather}
By moving all monomials to one side, we get the following equivalent
inequality:
\begin{gather}
2 + (-2) \cdot x + 1 \cdot x^2 > 0 \label{illustrative:ineq}
\end{gather}
A widely used criterion for validity of polynomial inequalities over the
natural numbers is
the \emph{absolute positiveness
  criterion}~\cite{hon:jak:98,con:mar:tom:urb:05}:
a polynomial $p(x_1,\dots,x_n)$ given as a sum of monomials is
\emph{absolutely positive} if the constant addend is greater than 0
and moreover the coefficients of all other monomials are greater than or
equal to 0. Absolute positiveness of $p(x_1,\dots,x_n)$ implies that
$p(x_1,\ldots,x_n) > 0$ holds for all
$x_1,\dots,x_n \in \mathbb{N}$~\cite{hon:jak:98}.\footnote{Contejean
et al.~\cite{con:mar:tom:urb:05} describe a version of absolute positiveness
for inequalities of the form $p(x_1,\ldots,x_n) \geq 0$. Here
the constant addend must be greater than \emph{or equal to} 0.}

In our example, absolute positiveness would mean that the following
conjunction of inequalities must hold:
$2 > 0 \land -2 \geq 0 \land 1 \geq 0$.
As the inequality $-2 \geq 0$ clearly does \emph{not} hold,
our polynomial is not absolutely positive, and
our attempt to prove termination with the help of the interpretation
$\AAA$ and the absolute positiveness criterion was unsuccessful.

But perhaps we can do better than absolute positiveness.
In the example, the inequality~\eqref{illustrative:ineq}
of the form $p(x) > 0$ can be split into two cases
for the universally quantified variable $x$,
by demanding equivalently that
$p(0) > 0$ holds and that $p(x+1) > 0$ holds for all
$x \in \mathbb{N}$ (i.e., $p(x) > 0$ must hold for $x \geq 1$).
In the example, the inequality $p(0) > 0$ boils down to $2 > 0$, which
clearly holds.
The inequality $p(x+1) > 0$ can be expressed as follows:
\begin{gather}
2 + (-2) \cdot (x+1) + 1 \cdot (x+1)^2 > 0
\end{gather}
By multiplying out and simplifying, we get:
\begin{gather}
1 + 1 \cdot x^2 > 0 \label{illustrative:ineq_one}
\end{gather}
For inequality \eqref{illustrative:ineq_one}, now absolute positiveness \emph{is}
applicable since the resulting inequalities indeed hold:
$1 > 0 \land 1 \geq 0$.
Thus, the polynomial interpretation $\AAA$ does induce a direct termination
proof for the rewrite rule \eqref{illustrative:rule} after all!
In the terminology of Hong and Jaku\v{s}~\cite{hon:jak:98}, the polynomial
$p(x)$ is absolutely positive \emph{from 1}: for all
$x \geq 1$,
the inequality $p(x) > 0$ holds over $\mathbb{N}$.

Interestingly, exhaustive search shows that there is no strictly monotonic
polynomial
interpretation to $\mathbb{N}$ with degree at most $2$ and coefficients
from $\{0,1\}$ for which absolute positiveness would allow us
to show validity of inequality \eqref{illustrative:polt} directly.
Thus, the case split as a preprocessing has allowed us to find a
direct termination
proof with a polynomial interpretation that was out of
reach with this search space for interpretations before.

\end{example}

\section{Automation}
\label{sec:automation}
In this section, I explain how to automate the approach sketched in
\autoref{sec:illustrative}.

The scenario is the following:
given a set of term constraints $\mathcal{C}$, we want to
find a (weakly or strictly monotonic) polynomial interpretation
for which all term constraints are fulfilled.

\begin{enumerate}
\item
Choose a parametric polynomial interpretation $\AAA$ that maps all $f \in \Sig$
to parametric polynomials $f_\AAA(x_1,\ldots,x_n)$
where $n$ is the arity of $f$. For example, for a unary symbol
$\Fs$, we might choose
$\polf{\AAA}{\Fs}(x) = a_0 + a_1 \cdot x + a_2 \cdot x^2$.
\item
Add constraints on parameters to enforce
well-definedness and monotonicity
of $\AAA$~\cite{con:mar:tom:urb:05,neu:mid:zan:10}.
For example, we could require
$a_0 \geq 0 \land a_1 \geq 0 \land a_2 \geq 0$
for well-definedness of $\polf{\AAA}{\Fs}$
and  $a_1 + a_2 \geq 1$ for strict monotonicity.
\item
Replace term constraints by polynomial constraints:
$s \succsim t$ becomes $[s]_\AAA \geq [t]_\AAA$, and
$s \succ t$ becomes $[s]_\AAA > [t]_\AAA$ (or equivalently
$[s]_\AAA \geq [t]_\AAA + 1$)
with implicit existential quantification over interpretation parameters
and universal quantification over all
$x_1, \dots, x_n \in \Var(s) \cup \Var(t)$, expressed as:
$p(x_1,\dots,x_n) \geq 0$
where $p(x_1,\dots,x_n) = p_0 + p_1 \cdot x_1^{e_{1,1}} \cdot \dots \cdot x_n^{e_{1,n}} + \dots +
 p_m \cdot x_1^{e_{m,1}} \cdot \dots \cdot x_n^{e_{m,n}}$,
all $p_i$ are now polynomials over interpretation parameters
$a_i$
(rather than just integer numbers), and
$x_i \in \Var(s) \cup \Var(t)$.
\item\label{step:elim}
Eliminate
 all
universally quantified variables
$x_i$,
get inequalities only over
parameters $a_i$.
\item
Ask a constraint solver for a solution for the interpretation parameters
over $\mathbb{Z}$ (if any),
replace the $a_i$ in $\AAA$
with their values from the solution
(e.g., $a_0 = a_1 = 1$ and $a_2 = 0$
yields $\polf{\AAA}{\Fs}$ in \autoref{sec:illustrative}).
\end{enumerate}

Step \ref{step:elim} often applies
the absolute positiveness criterion.
\autoref{sec:illustrative} shows that this criterion is incomplete.
As partial mitigation,
I propose providing an explicit value $\minX \in \mathbb{N}$
as input for the search ($\minX = 1$ as in \autoref{sec:illustrative}
is not necessarily sufficient). Then Step~\ref{step:elim}
replaces
\begin{align}
p(x_1,\dots,x_n) &\geq 0
\end{align}
by
\begin{align}
\bigwedge_{1 \leq i \leq n,\;
 q_i \in \{ 0, \dots, \minX - 1 \} \cup \{ x_i + \minX \}} p(q_1, \dots, q_n)
&\geq 0
\end{align}
and only then applies the conventional absolute positiveness criterion.

A
downside of this approach is that the number of inequalities
for a term constraint with $n$ variables can reach $(\minX+1)^n$.
However, while $n$ is in the exponent, often $n$ is not very big
in practice. For example, for constraints obtained from
string rewrite systems, we have $n = 1$ (although here the deep nesting
of terms may lead to polynomials with large exponents
and complex parametric polynomials after substitution of $x$ by $x+\minX$).

\section{Limitation}
\label{sec:limitation}

In this section, I discuss a limitation of the proposed case split
for validity checks of polynomials over the natural numbers.
The modified criterion becomes distinct from regular absolute
positiveness only if polynomials of degree at least $2$ are considered.
For a polynomial with only \emph{linear} monomials,
absolute positiveness is a prerequisite for positiveness.

\begin{theorem}
Let $p(x_1, \dots, x_n) = a_0 + a_1 x_1 + \dots + a_n x_n$
be a linear polynomial where $a_0, a_1, \dots, a_n \in \mathbb{Z}$.
Then $\forall x_1, \dots, x_n \in \mathbb{N}. p(x_1, \dots, x_n) > 0$
holds iff $a_0 > 0$ and $a_1, \ldots, a_n \geq 0$ hold.
\end{theorem}

\begin{proof}
``$\Leftarrow$'': The claim follows by absolute positiveness
of $p(x_1, \dots, x_n)$.

``$\Rightarrow$'':
Let $p(x_1, \dots, x_n)$ be of the above shape such that
$\forall x_1, \dots, x_n \in \mathbb{N}. p(x_1, \dots, x_n) \geq 0$.
Then $p(0,\dots,0) > 0$ implies $a_0 > 0$.

Now consider $i$ with $1 \leq i \leq n$. We have
$p(0, \dots, 0, x_i, 0, \dots, 0) = a_0 + a_i x_i$
(all other monomials are equal to $0$).
From $a_0 > 0$ and weak monotonicity of multiplication on $\mathbb{N}$,
we get that for $a_i \geq 0$ the statement
$\forall x_i \in \mathbb{N}. p(0, \dots, 0, x_i, 0, \dots, 0) > 0$
holds as required.
For $a_i < 0$, we would get the contradiction
$0 < p(0, \dots, 0, a_0+1, 0, \dots, 0)
= a_0 + a_i (a_0+1)
= a_0 (1 + a_i) + a_i
\leq a_i < 0$, so absolute positiveness is required.
\end{proof}

\section{Related Work}
\label{sec:related}

Neurauter, Middeldorp, and Zankl~\cite{neu:mid:zan:10} provide conditions
to ensure well-definedness and monotonicity of polynomial
interpretations to the
naturals of degree at most 3 where coefficients from $\mathbb{Z}$ are allowed.
These conditions also go beyond absolute positiveness,
but they consider the parametric polynomials assigned to function symbols
instead of compatibility with given term constraints
(i.e., whether an interpretation solves these term constraints).
From \cite[Remark~1]{neu:mid:zan:10}:

\begin{quote}
Since well-definedness of
a polynomial as defined above is equivalent to non-negativity of a polynomial
in $\mathbb{N}$, any method that ensures non-negativity of parametric polynomials can
also be used for checking compatibility. However, we remark that the method
presented in this paper is not ideally suited for this purpose as it also enforces
strict monotonicity, which is irrelevant for compatibility.
\end{quote}

For the question whether
$p(x_1, \dots, x_n) \geq 0$
is valid, monotonicity (strict or weak) of the polynomial function
$p(x_1, \dots, x_n)$ is indeed optional. For the example in
\autoref{sec:illustrative}, the criterion by Neurauter, Middeldorp, and Zankl
is not applicable, as the polynomial resulting from the term
constraint is indeed not monotonic: $p(0) = 2$, $p(1) = 1$, $p(2) = 2$.

\smallskip

The case analysis proposed in this extended abstract is also
related to the use of piecewise-defined monotone algebras in
termination proving.
Urban~\cite{urb:13} as well as Urban, Gurfinkel, and
Kahsai~\cite{urb:gur:kah:16} use piecewise-defined ranking
functions in the setting of imperative programs.
To prove termination of string rewrite systems,
Hofbauer~\cite{hof:18} proposes almost linear weight functions,
which consist of linear functions of the shape
$x + c$ for constant values $c$ combined with finitely
many exceptions for specific values of $x$.
Lucas and Guti\'errez~\cite{luc:gut:18} use piecewise definitions
for the automatic generation of models in an order-sorted
first-order setting with applications to proving termination
of term rewrite systems.
The approach of the present extended abstract is complementary
to these approaches:
its split into different cases can be applied independently
from the origin of the inequalities.

\section{Conclusion and Future Work}
\label{sec:conclusion}
I have presented a novel approach to improve the absolute positiveness
criterion
to determine validity of non-linear polynomial constraints.
The main motivation is in complexity analysis of term rewrite systems~\cite{hir:mos:08,nos:emm:gie:13},
where interpretations to non-linear polynomials of degree $m$
can induce
time complexity bounds of degree $m$. In this light, the illustrative
example in \autoref{sec:illustrative} is perhaps slightly less artificial
than it appears at first glance: for example, for a rewrite rule
$\Ff(x, y) \to \Fcons(\Fg(\Fh(x), y), \Fi(\Fs(x)))$ where $\Ff$,
$\Fg$, $\Fh$, and $\Fi$ are the defined symbols, the
\emph{Dependency Tuple Framework}~\cite{nos:emm:gie:13}
leads to a term constraint of the form
$\Ff^\sharp(x, y) \succ \Fc(\Fg^\sharp(\Fh(x), y), \Fh^\sharp(x), \Fi^\sharp(\Fs(x)))$
where $\Fc$ is a fresh constructor symbol.
Here a polynomial interpretation is needed that assigns a polynomial $x_1 + \dots + x_n$ to $\Fc$ (similar to \autoref{sec:illustrative}) and
uses non-linear interpretations as an explicit feature to find non-linear
complexities. It is worth noting that here nested function calls lead to
repetitions of subterms (and variables) on right-hand sides, similar
to our (artificial) repetition of the variable $x$ in the example.

An obvious next step would be to implement the proposed new preprocessing for absolute
positiveness in a termination and complexity analysis tool (e.g.,
\aprove~\cite{gie:asc:bro:emm:fro:fuh:hen:ott:plu:sch:str:swi:thi:17})
that searches for polynomial interpretations.
In the first instance, the parameter $\minX$ would be given by the user.
Finding sufficient criteria to determine upper bounds for $\minX$
beyond which a further case split provably does not lead to further power
would be future work. Orthogonally, heuristics to determine $\minX$ on
a per-variable basis rather than uniformly for all variables could mitigate
the combinatorial explosion mentioned at the end of \autoref{sec:automation}.

Finally, experiments on a large benchmark set such as the Termination
Problem DataBase (TPDB)\footnote{\url{https://github.com/TermCOMP/TPDB-ARI/}}
used for the annual Termination Competition~\cite{gie:rub:ste:wal:yam:19}
would provide insight into the practical usefulness of the approach.



\bibliography{references.bib}

\end{document}